\documentclass[conference,a4paper]{IEEEtran}
\usepackage{epsfig}
\usepackage{times}
\usepackage{float}
\usepackage{afterpage}
\usepackage{amsmath}
\usepackage{mathrsfs}
\usepackage{amstext}
\usepackage{amssymb,bm}
\usepackage{latexsym}
\usepackage{color}
\usepackage{graphicx}
\usepackage{amsmath}
\usepackage{amsthm}
\usepackage{graphicx}
\usepackage[center]{caption}
\usepackage{pstricks}
\usepackage{subfigure}
\usepackage{booktabs}
\usepackage{multicol}
\usepackage{lipsum}
\usepackage{dblfloatfix}
\usepackage{algorithm} 
\usepackage[noend]{algpseudocode} 
\usepackage[normalem]{ulem}
\usepackage{enumitem}
\usepackage{bbm}

\newcommand{\mcal}{\mathcal}
\newcommand{\msf}{\mathsf}

\newtheorem{thm}{Theorem}

\newtheorem{lem}[thm]{Lemma}

\newtheorem{rem}{Remark}
\newtheorem{defin}{Definition}
\newtheorem*{defin*}{Definition}
\newtheorem{prope}{Property}
\algnewcommand\algorithmicforeach{\textbf{for each}}

\algdef{S}[FOR]{ForEach}[1]{\algorithmicforeach\ #1\ \algorithmicdo}

\input{widebar_hack.tex}

\begin{document}

\title{ Efficiently Finding Simple Schedules in Gaussian Half-Duplex Relay Line Networks} 
\author{
\IEEEauthorblockN{Yahya H. Ezzeldin$^\dagger$, Martina Cardone$^{\dagger}$, Christina Fragouli$^{\dagger}$, Daniela Tuninetti$^*$}
$^{\dagger}$ UCLA, Los Angeles, CA 90095, USA,
Email: \{yahya.ezzeldin, martina.cardone, christina.fragouli\}@ucla.edu\\
$^*$ University of Illinois at Chicago,
Chicago, IL 60607, USA, 
Email: danielat@uic.edu
}
\IEEEoverridecommandlockouts
\maketitle

\begin{abstract}
The problem of operating a Gaussian Half-Duplex (HD) relay network optimally is challenging due to the exponential number of listen/transmit network states that need to be considered.
Recent results have shown that, for the class of Gaussian HD networks with $N$ relays, there always exists a {\it simple} schedule, i.e., with at most $N+1$ active states, that is sufficient for approximate (i.e., up to a constant gap) capacity characterization.
This paper investigates how to efficiently find such a simple schedule over 
line networks.
Towards this end, a polynomial-time algorithm is designed and proved to output a simple schedule that achieves the approximate capacity.
The key ingredient of the algorithm is to leverage similarities between network states in HD and edge coloring in a graph.
It is also shown that the algorithm allows to derive a closed-form expression for the approximate capacity of the Gaussian line network that can be evaluated distributively and in linear time.
Additionally, it is shown using this closed-form that the problem of Half-Duplex routing is NP-Hard.
\end{abstract}
\section{Introduction}
Computing the capacity of a wireless relay network is a long-standing open problem. 
For Half-Duplex (HD) networks, where the $N$ relays cannot simultaneously transmit and receive, 
such problem is more challenging due to the $2^N$ possible listen/transmit configuration states that need to be considered.
Recently, in~\cite{CardoneITW2015} the authors proved the conjecture posed in~\cite{BrahmaIT2016}, which states that {\it simple} schedules (i.e., with at most $N+1$ active states) suffice for approximate (i.e., up to a constant gap) capacity characterization.
This result is promising as it implies that the network can be operated close to its capacity with a limited number of state switches.
However, to the best of our knowledge, it is not clear yet if such simple schedules 
can be found efficiently.

The main result of this work is an algorithm design that allows to compute a
simple schedule that achieves the approximate capacity of the $N$-relay Gaussian HD line network with complexity $O(N^2)$.
The algorithm leverages similarities between network states in HD and edge coloring in a graph, by associating different colors to links that cannot be activated simultaneously.
In addition, the algorithm allows to derive the approximate capacity of the Gaussian HD line network in closed form. 
This expression has two appealing features: (i) it can be evaluated in linear time and (ii) it can be distributively computed among the $N$ relays.
The novelty and applicability of the results presented in this paper are two-fold:
(i) they shed light on how to operate a class of Gaussian HD relay networks close to the capacity with the minimum number of state switches and (ii) they represent the first approximate capacity characterization in closed form for a class of Gaussian HD relay networks with general number of relays.

%
\smallskip

\noindent{\bf{Related Work.}} 
The capacity of the $N$-relay Gaussian HD network is not known in general. 
Recent results in~\cite{AvestimehrIT2011,OzgurIT2013} showed
that the capacity can be approximated to within a constant gap by the cut-set upper bound evaluated with independent inputs and a schedule, which is independent of the transmitted and received signals. 
In the rest of the paper, we refer to this bound as the {\it approximate capacity}.
To the best of our knowledge, the tightest known gap for Gaussian HD relay networks is of $1.96 (N+2)$ bits per channel use, independently of the channel parameters~\cite{CardoneIT2014}.

In general, the evaluation of the approximate capacity is cast as an optimization problem over
$2^N$ listen/transmit states.
As $N$ increases, this evaluation, as well as determining an optimal schedule of listen/transmit states, become computationally expensive. 
The authors in~\cite{OngIT2012} designed an iterative algorithm to determine an  approximately optimal schedule when the relays use decode-and-forward.
In~\cite{EtkinIT2014}, the authors proposed a `grouping' technique to address the complexity of the aforementioned optimization problem. This technique allows to compute the approximate capacity in polynomial-time for certain classes of  Gaussian HD relay networks that include the line network as special case.
While the results in~\cite{OngIT2012} and~\cite{EtkinIT2014} show that the approximate capacity can be efficiently obtained 
for special network topologies by solving a linear program, it is not clear how to construct (in polynomial time) a schedule that achieves the approximate capacity.
Differently, in this work we design a polynomial-time algorithm that outputs a simple schedule, which achieves the approximate capacity of Gaussian HD line networks and allows to compute this quantity in closed form.


\smallskip

\noindent{\bf{Paper Organization.}} 
Section~\ref{sec:model} describes the $N$-relay Gaussian HD line network and presents known capacity results.
Section~\ref{sec:mainresults} discusses our main results and their implications.
Section~\ref{sec:cuts} simplifies the approximate capacity expression for Gaussian HD line networks.
Section~\ref{sec:algo} designs an algorithm for finding a simple schedule for a Gaussian HD line network that achieves the approximate capacity.
Section~\ref{sec:conc} concludes the paper. Some of the proofs 
are delegated to the Appendix.
Particularly, Appendix~\ref{HD_NP_hard} proves that the problem of find the best Half-Duplex route in a network is NP-Hard.

\section{System Model}
\label{sec:model}
We consider the $N$-relay Gaussian HD line network $\mcal{L}$
where a source node (node $0$) wishes to communicate to a destination node (node $N+1$) through a route of $N$ relays where each relay is operating in HD.
The input/output relationship for the line network $\mcal{L}$
is
\begin{align}
    Y_i & = \left (1-S_i \right ) h_{i,i-1} X_{i-1} S_{i-1} + Z_i, \ \forall i \in [1:N+1],
\label{eq:inputoutput}
\end{align}
where: 
(i) $X_i$ (respectively, $Y_i$) denotes the channel input (respectively, output) at the $i$-th node; 
(ii) $S_i$ is the binary random variable which represents the state 
of node $i$, i.e., if $S_i=0$ then node $i$ is receiving, while if $S_i=1$ then node $i$ is transmitting;
notice that $S_0=1$ (i.e., the source always transmits) 
and $S_{N+1}=0$ (i.e., the destination always receives);
(iii) $Z_i$ indicates the additive white Gaussian noise at node $i$, where the noises are assumed to be independent and identically distributed as $\mcal{CN}(0,1)$; (iv) $h_{i,j}$ denotes the complex channel coefficient from node $j$ to node $i$ and $h_{i,j}=0$ whenever $j \neq i-1$; the channel gains are assumed to be constant for the whole transmission duration and hence known to all nodes;
(v) the channel inputs satisfy an average power constraint $\mathbb{E}[|X_i|^2]\leq 1,\ \forall i \in [0:N]$. 
We denote
the point-to-point link capacity from node $i-1$ to node $i$ with
\begin{align*}
    \ell_i = \log\left(1+|h_{i,i-1}|^2 \right),\quad \forall i \in [1:N+1].
\end{align*}
The capacity of the Gaussian HD line network $\mcal{L}$ described in \eqref{eq:inputoutput}
can be approximated to within a constant gap $G=O(N)$~\cite{AvestimehrIT2011,OzgurIT2013,LimIT2011,KramerAllerton2004}\footnote{In~\cite{KramerAllerton2004}, it was observed that information can be conveyed by randomly switching the relay between transmit and receive modes. However, this only improves the capacity by a constant, at most $1$ bit per relay.
}
%
\begin{align}
\label{eq:ApproxCap_orig}
{\mathsf{C}}_\mcal{L}= \max_{\lambda \in \Lambda} \min_{\mathcal{A} \subseteq [1:N]} \sum_{s \in [0:1]^N} \lambda_s 
\sum_{\substack{i \in  \{N+1\} \cup \left \{\mathcal{R}_s \cap \mathcal{A} \right \}  \\ i-1 \in \{0\} \cup \left \{ \mathcal{T}_s \cap \mathcal{A}^c \right \} } } \ell_{i},
\end{align}
where: 
(i) the schedule $\lambda \in \mathbb{R}^{2^N}$ determines the fraction of time the network operates in each of the states $s \in [0{:}1]^N$, i.e., $\lambda_s = \Pr \left (S_i = s_i, \forall i \in [1:N]  \right )$;
(ii) $\Lambda = \{ \lambda : \lambda \in \mathbb{R}^{2^N},\ \lambda \geq 0,\ \|\lambda\|_1 = 1 \}$ is the set of all possible schedules; 
(iii) $\mathcal{R}_s$ (respectively, $\mathcal{T}_s$) represents the set of indices of relays receiving (respectively, transmitting) in the relaying state $s \in [0{:}1]^N$; 
(iv) $\mathcal{A}^c = [1{:}N]\backslash \mathcal{A}$.
%

\noindent We can equivalently write
the expression in \eqref{eq:ApproxCap_orig} as
\begin{align}
\label{eq:ApproxCap}
{{\mathsf{C}}_\mcal{L}}= \max_{\lambda \in \Lambda} \min_{\mathcal{A} \subseteq [1:N]} \sum_{s \in [0:1]^N} \lambda_s 
\sum_{\substack{i \in  \{N+1\} \cup\mathcal{A}  \\ i-1 \in \{0\} \cup \mathcal{A}^c  } } \ell_{i,s}',
\end{align}
\begin{align}
\label{eq:ChannelState}
\hspace{-1em}\text{where} \ \ell_{i,s}' {=}
\begin{cases}
    \ell_{i,} & \text{if} \ i \!\in \! \mathcal{R}_s {\cup} \{N{+}1\} \ \text{and} \ i{-}1 \! \in \! \mathcal{T}_s\cup\{0\}\\
    0, & \text{otherwise.}
\end{cases}
\end{align}
Similarly, we denote with ${\mathsf{C}}^{\lambda}_\mcal{L}$, the HD rate achieved by the line network $\mcal{L}$ when operated with the fixed schedule $\lambda$, i.e.,
\begin{align}
\label{eq:ApproxCap_fixedschedule}
{{\mathsf{C}}^{\lambda}_\mcal{L}}= \min_{\mathcal{A} \subseteq [1:N]} \sum_{s \in [0:1]^N} \lambda_s 
\sum_{\substack{i \in  \{N+1\} \cup\mathcal{A}  \\ i-1 \in \{0\} \cup \mathcal{A}^c  } } \ell_{i,s}'.
\end{align}
Note that for all possible schedules $\lambda$, $\mathsf{C}^{\lambda}_\mathcal{L} \leq \mathsf{C}_\mathcal{L}$.

\begin{defin}{\rm 
    We say that a schedule $\lambda$ is \emph{simple} if the number of active states (i.e., states $s$ such that $\lambda_s\neq 0$) is less than or equal to $N+1$. 
In other words, we have
$\| \lambda \|_0 \leq N+1$.
}
\end{defin}

\section{Main Results and Discussion}
\label{sec:mainresults}
Our main result, stated in the theorem below, is two-fold:
(i) it designs a polynomial-time algorithm that outputs a simple schedule optimal for approximate capacity and
(ii) it provides a closed-form expression for the approximate capacity of the HD line network that can be evaluated in linear time.
\begin{thm}
\label{thm:closed_form}
For the $N$-relay Gaussian HD line network $\mcal{L}$ described in~\eqref{eq:inputoutput}, 
a simple schedule optimal for approximate capacity
can be obtained in $O(N^2)$ time and the approximate capacity $\msf{C}_\mcal{L}$ in~\eqref{eq:ApproxCap} is given by
    \begin{align}
\label{eq:closed_form}
        {\msf{C}_\mcal{L}} = \min_{i \in [1:N]}\left\{\frac{\ell_i\ \ell_{i+1}}{ \ell_i+\ell_{i+1}} \right\}.
    \end{align}
\end{thm}
\noindent {\bf Converse.} It is not hard to argue that the right-hand side of~\eqref{eq:closed_form} is an upper bound on $\msf{C}_\mcal{L}$. 
This can be seen by assuming that, for a given $i \in [1:N]$, node $i-1$ perfectly cooperates with node $0$ and similarly node $i+1$ perfectly cooperates with node $N+1$. 
Clearly, the capacity of this new line network is an upper bound on $\msf{C}_\mcal{L}$ and it has an approximate capacity equal to $\frac{\ell_i \ell_{i+1}}{\ell_i + \ell_{i+1}}$. 
Since this is true for all $i \in [1:N]$, then $\msf{C}_\mcal{L}$ is less than or equal to the right-hand side of~\eqref{eq:closed_form}.
The heart of the proof is thus to prove the achievability of~\eqref{eq:closed_form}.

Before we delve into the proof of the achievability in Theorem~\ref{thm:closed_form},
we highlight the following remarks to motivate the need to search for a simple schedule for the line network.

\begin{rem}
{\rm {\bf Are two active states sufficient for approximate capacity characterization?}
Consider a line network with one relay.
For this network, the schedule that achieves the approximate capacity
has only two active states, which activate the links alternatively.
Intuitively, one might think that this would extend to general number of relays.
For example, for a network with $N=3$, can we achieve the approximate capacity by only considering the listen/transmit states $s_1 = 010,\ s_2 = 101$? 
Surprisingly, the answer to this question is negative as we illustrate through the following example with $N=3$ and
\begin{align}
\label{eq:RunEx}
\ell_1 = 2R,\ \ell_2 = 2R,\  \ell_3 =3R,\ \ell_4 = R,
\end{align}
where $R > 0$.
By considering only the two aforementioned states, we can achieve a rate of $2$/$3 R$. 
However, in Section~\ref{sec:algo}, we consider this network as our running example and show that using $N+1 =4$ states, we can achieve $3$/$4R$.}
\end{rem}

\begin{rem}
{\rm {\bf Can we a priori limit our search over a polynomial number of states?}
For the Full-Duplex (FD) line network, we can a priori limit our search for the minimum cut over $N+1$ cuts (instead of $2^N$).
This reduction in the number of cuts is also possible for the HD line network as we prove in the next section. 
This fact raises the question whether we can also a priori reduce the search space for the active states to a polynomial set (instead of $2^N$). 
This is not possible as we state in the theorem below, which is proved in Appendix~\ref{app:ProofofNPhard}. 
\begin{thm}
\label{thm:NPhard}
With only the knowledge that relays are arranged in a line, the cardinality of the smallest search space of states over which a schedule optimal for approximate capacity can be found is $\Omega(2^{N/3})$.
\end{thm}
}
\end{rem}

\begin{rem}{\rm Theorem~\ref{thm:closed_form} has two appealing consequences:

{1)} The capacity of the line network with $N$ relays can be computed in linear time in $N$. 
This improves on the result in~\cite{EtkinIT2014}, where the approximate capacity can be found in polynomial time (but not linear in the worst case) by solving a linear program with $O(N)$ variables.
%

{2)}  The approximate capacity in Theorem~\ref{thm:closed_form} can be computed in a distributive way as follows. Each relay $i \in [1:N]$ computes the quantity
\begin{align*}
m_i= \min \left \{ \frac{\ell_{i}\ell_{i+1}}{\ell_{i}+\ell_{i+1}}, m_{i-1}\right \},
\end{align*}
where $m_{0}=\infty$, and sends it to relay $i+1$.
With this, at the end we have $m_N={\mathsf{C}}_\mcal{L}$.
In other words, for approximate capacity computation, it is only required that each relay knows the capacity of the incoming and outcoming links.
}
\end{rem}


\section{Fundamental Cuts in HD Line Networks}
\label{sec:cuts}
In this section, we prove that in a Gaussian HD line network, we can compute the approximate capacity $\msf{C}_\mcal{L}$ in~\eqref{eq:ApproxCap} by considering only $N+1$ cuts, which are the same that one would need to consider if the network was operating in FD.

For the Gaussian line network $\mcal{L}$,
when all the $N$ relays operate in FD, the FD capacity is given by
\begin{align}
\label{eq:ApproxCapFD}
{\mathsf{C}}^{{\rm{FD}}}_\mcal{L} 
= \min_{\mcal{A}\subseteq [1:N]} \sum_{\substack{i \in \{N+1\}\cup \mcal{A},\\i-1 \in \{0\}\cup \mcal{A}^c}} \ell_{i}
= \min_{i \in [1:N+1]} \left \{ \ell_{i}\right \},		
\end{align}
that is, without {\it explicit} knowledge of the values of $\ell_{i}$ or their ordering, the number of cuts over which we need to optimize (see ${\mathsf{C}}^{{\rm{FD}}}_\mcal{L}$ in~\eqref{eq:ApproxCapFD}) is $N+1$.
We refer to these cuts as {\it fundamental}. Let $\mathscr{F}$ denote the set of these fundamental cuts (which are of the form $\mathcal{A}=[i:N], i \in [1:N]$ or $\mathcal{A}=\emptyset$), then for any cut $\mcal{A}$ of the network, there exists a fundamental cut $F(\mcal{A}) \in \mathscr{F}$ such that:
\begin{align}
\label{fundamental_cuts_eq}
\sum_{\substack{i \in \{N+1\}\cup F(\mcal{A}),\\i-1 \in \{0\}\cup F(\mcal{A)}^c}} \ell_{i} \leq 
\sum_{\substack{i \in \{N+1\}\cup \mcal{A},\\i-1 \in \{0\}\cup \mcal{A}^c}} \ell_{i}.
\end{align}
Furthermore, the function $F(\cdot)$ does not depend on the values of $\ell_i$. 
We next prove that the fundamental cuts in HD equal those in~\eqref{eq:ApproxCapFD} for FD. 
Consider a fixed 
schedule $\lambda$. Then, by using~\eqref{fundamental_cuts_eq} for the inner summation in~\eqref{eq:ApproxCap_fixedschedule}, for each $s\in[0:1]^N$ we have
\begin{align*}
\sum_{\substack{i \in  \{N+1\} \cup  F(\mathcal{A}),  \\ i-1 \in \{0\} \cup F(\mathcal{A})^c  } } \ell_{i,s}'
\leq \sum_{\substack{i \in  \{N+1\} \cup  \mathcal{A},  \\ i-1 \in \{0\} \cup \mathcal{A}^c  } } \ell_{i,s}' 
.
\end{align*}
Thus, we can simplify \eqref{eq:ApproxCap_fixedschedule} as
\begin{align}
\label{eq:simplified_expression_cap}
{{\mathsf{C}}^{\lambda}_\mcal{L}} &= \min_{\mathcal{A} \subseteq [1:N]} {\sum_{s \in [0:1]^N}} \lambda_s 
\sum_{\substack{i \in  \{N+1\} \cup \mathcal{A}  \\ i-1 \in \{0\} \cup \mathcal{A}^c  } } \ell_{i,s}' \nonumber\\
&= \min_{\mathcal{A} \in \mathscr{F}} {\sum_{s \in [0:1]^N}} \lambda_s
\sum_{\substack{i \in  \{N+1\} \cup \mathcal{A}  \\ i-1 \in \{0\} \cup \mathcal{A}^c  } } \ell_{i,s}' \nonumber\\
&=\min_{i \in [1:N+1]} \left( \sum_{\substack{s \in \mcal{S}_i}} \lambda_s \right) \ell_{i},
\end{align}
where
\begin{align*}
\mcal{S}_i =& \{s \in [0:1]^N | i \in \{N+1\} \cup \mcal{R}_s ,\ i-1 \in \{0\} \cup \mcal{T}_s \}.
\end{align*}
The set $\mcal{S}_i \subseteq [0:1]^N$ represents the collection of states that activate the $i$-th link. 
For illustration consider a network with $N = 3$. We have
\begin{align*}
    &\mcal{S}_1 = \{000,\ 001,\ 010,\ 011\},& & \mcal{S}_2 = \{100,\ 101\},  \\
    &\mcal{S}_3 = \{010,\ 110\},& & \mcal{S}_4 = \{001,\ 011,\ 101,\ 111\}.
\end{align*}
Using the same arguments as in~\eqref{eq:simplified_expression_cap}, we can similarly simplify the expression of $\msf{C}_\mcal{L}$ in~\eqref{eq:ApproxCap}. 
Thus, the result presented in this section explicitly provides the $N+1$ cuts (out of the $2^N$ possible ones) over which it is sufficient to minimize in order to obtain $\msf{C}_\mcal{L}$ in~\eqref{eq:ApproxCap}.

\section{Finding a Simple Schedule Optimal for Approximate Capacity}
\label{sec:algo}
In this section, we design a polynomial-time algorithm that finds a simple schedule, which achieves the approximate capacity of the
$N$-relay Gaussian HD line network.
The algorithm leverages similarities between network states in HD and edge coloring in a graph. 
In particular, an edge coloring assigns colors to edges in a graph such that no two adjacent edges are colored with the same color. Similarly in HD, a network state cannot be a receiver and a  transmitter simultaneously. Thus, if we assign one color to all activated links (viewed as edges) in a network state, this does not violate the rules of edge coloring in a graph. 
In what follows, we first explain how the algorithm makes use of these similarities assuming that the link capacities $\ell_i$ are all integers and later in the section, we show how the algorithm extends to rational and real values of $\ell_i$.

\subsection{An Algorithm for Networks with Integer Link Capacities}
\label{subsec:algo_integer}
%
%
%
%
We highlight the algorithm procedure in the following main steps and provide intuitions for each step.
As a {\it running example} to illustrate the different steps we consider the line network $\mathcal{L}$ with $N=3$ relays described in~\eqref{eq:RunEx} with $R = 1$.

\noindent{\bf{Step 1.}}
Let $M$ be a common multiple of the link capacities $\ell_{i}$.
For the line network $\mcal{L}$ we construct an associated graph $G_\mcal{L}$ where: (i) the set of nodes is the same as in the network $\mcal{L}$ and (ii)
each link with capacity $\ell_i$ in $\mcal{L}$ is replaced by $n_i$ parallel edges, where
\begin{align}
\label{eq:n_i}
n_i = \frac{M}{\ell_{i}},\quad \forall i \in [1:N+1].
\end{align}
Clearly, computing $M$ and $n_i$ requires $O(N)$ operations.
The main motivation behind this step is that from~\eqref{eq:simplified_expression_cap}, it is not difficult to see that a good schedule would try to assign more weights $\lambda_s$ to a link with a smaller capacity. \
Hence, if we treat edge colors as equally weighted, a link with a smaller capacity should get more colors. 
Thus, the main idea above is to assign $n_i$ adjacent edges inversely proportional to $\ell_i$.

\noindent {\it Running Example.}
We have
\begin{align*}
M = 6 \quad
\text{and}\ \ 
n_1 = 3, \ \
n_2 = 3, \ \
n_3 = 2, \ \
n_4 = 6.
\end{align*}

\noindent{\bf{Step 2.}}
In this step, our goal is to edge color the graph $G_\mcal{L}$. By noting that $G_\mcal{L}$ is a bipartite graph, we know that an optimal coloring can be performed with $\Delta$ colors, where $\Delta$ is the maximum node degree and is equal to
\begin{align}
\label{eq:degree}
\Delta = \max_{i \in [1:N]} \left\{ n_i + n_{i+1}\right\}.
\end{align}
In particular,
we define our coloring by the interval of colors $\mcal{C}_i \subset [1:\Delta]$ that are assigned to the $n_i$ edges that connect node $i-1$ to node $i$ such that $|\mcal{C}_i | = n_i$. 
Specifically, we assign $\mcal{C}_i$ for $i \in [1:N{+}1]$ as
\begin{align}
\mcal{C}_i = 
\begin{cases}
[1:n_i],  & i\ \text{ even}\\
[\Delta{-n_i+}1:\Delta], & i\ \text{ odd}.
\end{cases}
\end{align}
The complexity of this step is $O(N)$, since for each $i \in [1:N+1]$ we only compute two numbers that define the interval $\mcal{C}_i$, that is the two limit points $\mcal{C}_i^{(\ell)}$ and $\mcal{C}_i^{(r)}$ of the interval, i.e., $\mcal{C}_i=[\mcal{C}_i^{(\ell)}:\mcal{C}_i^{(r)}]$.

\noindent {\it Running example.}
Since we have $\Delta=8$, then the assigned color intervals are
\begin{align*}
\mcal{C}_1 = [6:8], \ \
\mcal{C}_2 = [1:3], \ \
\mcal{C}_3 = [7:8], \ \
\mcal{C}_4 = [1:6].
\end{align*}
\noindent{\bf{Step 3.}}
From the previous step, we have $\Delta$ colors each of which corresponds to a network state running for $1/\Delta$ fraction of time. 
However, some of these colors can represent the same operation states. 
For instance, in our {\it running example}, the colors 7 and 8 appear in both $\mcal{C}_1$ and $\mcal{C}_3$ (and nowhere else). 
Therefore, we can group the time fractions of colors 7 and 8 together, since they operate the network in the same way. 
To perform this color grouping, we run an iterative algorithm over the color intervals $\mcal{C}_i, i \in[1:N+1]$ constructed in the previous step, which outputs a schedule for the network. 
The algorithm pseudocode is shown in Algorithm~\ref{algo:interval_assign} and can be summarized as follows:

\smallskip

\noindent {\bf{1)}} We first find a descendingly ordered set of colors $p_u$ at which the network state changes. The network state changes whenever an interval $\mcal{C}_i$ begins or ends. 
Therefore in this step, we sort the unique elements of an array $p$ that contains the $\mcal{C}_i^{(\ell)}$ and $\mcal{C}_i^{(r)}+1$ for all $i \in [1:N+1]$.
Since there are $N+1$ different $\mcal{C}_i$ then this operation takes $O(N\log N)$ time.
At this point, it is worth noting that for odd $i$, we have $\mcal{C}_i^{(r)} = \Delta$ while for even $i$, we have $\mcal{C}_i^{(\ell)} = 1$. 
Thus,  the descendingly ordered array $p_u$ has at most $N+2$ unique elements. 

\noindent {\it Running example.} We have ${p}_u = \{9,\ 7,\ 6,\ 4,\ 1\}$.

\smallskip

\noindent {\bf{2)}} Next, we go through the array $p_u$ to construct the group of colors $I_j$. Each set $I_j$ represents a network state and all colors in a set $I_j$ operate the network in the same way.
        To do this we compute the endpoints $I_j^{(\ell)},\ I_j^{(r)}$ for each $I_j$. The fraction of time the network operates in the state represented by $I_j$ is stored in the vector $w$ and is calculated as $|I_j| / \Delta$.
        
\smallskip

\noindent {\bf{3)}} For each $I_j$ constructed, the algorithm performs a loop of $N$ iterations in order to determine the state of each node in this particular network configuration state and records this in a row of the matrix $\Lambda$. 
Thus, the active states are represented by rows of $\Lambda$. The algorithm finally outputs the variables $\Lambda$ and $w$. The complexity of steps 2 and 3 is $O(N^2)$.

\noindent {\it Running example.}
Algorithm~\ref{algo:interval_assign} outputs
\begin{align*}
    &I_1 = [7:8],\ \ \ \Lambda(1,:) = 010,\ \ \ w(1) = 2/8,\\
    &I_2 = [6:6],\ \ \ \Lambda(2,:) = 001,\ \ \ w(2) = 1/8,\\
    &I_3 = [4:5],\ \ \ \Lambda(3,:) = 111,\ \ \ w(3) = 2/8,\\
    &I_4 = [1:3],\ \ \ \Lambda(4,:) = 101,\ \ \ w(4) = 3/8.
\end{align*}
Finally, we note that since $p_u$ has at most $N+2$ terms, then the number of states output by the algorithm is at  most $N+1$, i.e., the schedule is simple as also stated in Theorem~\ref{thm:closed_form}.
%
\begin{algorithm}
\caption{Grouping Edge Colors}
\label{algo:interval_assign}
 {
 \begin{algorithmic}
     \State \textbf{Input} $N, {\Delta}, (\mcal{C}_i^{(\ell)}, \mcal{C}_i^{(r)}), \forall i \in [1:N+1]$
     \State \textbf{Output} $\Lambda, w$
   \ForEach {$i \in [1:N+1]$}
   \State $p(i) \leftarrow \mcal{C}_i^{(\ell)}, \quad p(i{+}N{+}1) \leftarrow  \mcal{C}_i^{(r)}+1 $
        
\EndFor
\State $p_u \leftarrow UniqueDescSort(p)$
\ForEach {$j \in [1: length(p_u)-1]$}
\State $I_j^{(r)} \leftarrow p_u(j)-1,\quad I_j^{(\ell)} \leftarrow p_u(j+1)$
\State $w(j) \leftarrow (I_j^{(r)}- I_j^{(\ell)} + 1)/\Delta, \quad state \leftarrow 1$
\ForEach {$i \in [1:N+1]$}
\If {  $[I_j^{(\ell)}:I_j^{(r)}] \subseteq \mcal{C}_i$}
\State $state \leftarrow 0$
\If {$i < N+1$} $\Lambda(j,i) \leftarrow 0$ \EndIf
\If {$i > 1$} $\Lambda(j,i-1) \leftarrow 1$ \EndIf
\ElsIf {$i<N+1$}
\State $\Lambda(j,i) \leftarrow state$
\EndIf
\EndFor
\EndFor
 \end{algorithmic}
 }
 \end{algorithm}

\subsection{Rate Achieved by the Schedule}
\label{subsection:closedform}
In the previous subsection we developed an algorithm that outputs a simple schedule for the $N$-relay Gaussian HD line network $\mcal{L}$. We now derive the rate that the constructed schedule achieves. 
Let $\lambda$ be the schedule output by Algorithm~\ref{algo:interval_assign}, then we have the following relation $\forall i \in [1:N+1]$
\begin{align}
\label{eq:sum_lambda_ni_Delta}
\sum_{s \in \mathcal{S}_i} \lambda_s= \sum_{\substack{j \in [1:K],\\ \Lambda(j,:) \in \mathcal{S}_i}} w(j)
= \sum_{\substack{j \in [1:K],\\ \Lambda(j,:) \in \mathcal{S}_i}} \frac{|I_j|}{\Delta} = \frac{n_i}{\Delta},
\end{align}
where $K$ is the number of network configuration states output by Algorithm~\ref{algo:interval_assign} and $\mathcal{S}_i$ is the set of states which make the link of capacity $\ell_{i}$ active and is defined in~\eqref{eq:simplified_expression_cap}. 
From~\eqref{eq:simplified_expression_cap},~\eqref{eq:n_i} and~\eqref{eq:sum_lambda_ni_Delta}, we can simplify the achievable rate $ {\mathsf{C}}^{\lambda}_\mcal{L}$ as
\begin{align}\label{eq:rate_achieved_1}
{\mathsf{C}}^{\lambda}_\mcal{L} &= \min_{i \in [1:N+1]} \left( \sum_{\substack{s \in \mcal{S}_i}} \lambda_s \right) \ell_{i} 
\stackrel{\eqref{eq:sum_lambda_ni_Delta}}= \min_{i \in [1:N+1]} \frac{n_i}{\Delta} \ell_{i} \stackrel{\eqref{eq:n_i}}= \frac{M}{\Delta}.
\end{align}
From the definition of maximum degree $\Delta$ in~\eqref{eq:degree}, we obtain
\begin{align}
\label{eq:Delta_M}
\frac{\Delta}{M} &= \frac{1}{M}\max_{i \in [1:N]} \left \{ 
n_i+n_{i+1}\right \} 
{=} \max_{i \in [1:N]}  \left \{ \frac{1}{\ell_{i}} {+} \frac{1}{\ell_{i+1}}
\right \}.
\end{align}
Finally, by substituting~\eqref{eq:Delta_M} into~\eqref{eq:rate_achieved_1} we obtain that the simple schedule $\lambda$ constructed from 
our proposed polynomial-time algorithm
achieves the right-hand side of~\eqref{eq:closed_form}.
This concludes the proof of Theorem~\ref{thm:closed_form} when the link capacities are integers.
\subsection{Extension to Real Link Valued Capacities}
\label{sub:IrrLinkCap}
The extension of the algorithm in Section \ref{subsec:algo_integer} to networks with rational link capacities is straightforward, by simply multiplying each of the link capacities with a common multiple of the denominators.
In this subsection, we focus on how the algorithm can be used to find a schedule for the line network $\mathcal{L}$ with real link capacities such that the rate achieved by this schedule is at most a constant gap $\varepsilon$ away from $\mathsf{C}_{\mathcal{L}}$.

For a fixed $\varepsilon >0 $, let $\mathcal{L}_{q,\varepsilon}$ be a line network with rational link capacities $q_{i}$ such that
\begin{align}
\label{eq:rat_is_dense}
\forall i \in[1:N+1],\quad  q_{i} \in \mathbb{Q},\  \ell_{i} - \varepsilon \leq q_{i} \leq \ell_{i}.
\end{align}
Such $\mathcal{L}_{q,\varepsilon}$ always exists (but it is not unique) since the set of rationals $\mathbb{Q}$ is dense in $\mathbb{R}$. 
We now relate $\msf{C}_{\mcal{L}_{q,\varepsilon}}$ and $\msf{C}_{\mcal{L}}$ by appealing to the lemma below, which we prove in Appendix~\ref{app:ProofLemmaRealtoRat}.
\begin{lem}
    Let $\mathcal{L}$ be a line network with real link capacities $\ell_{i}$, then $\forall \varepsilon >0$, we have
\begin{align}
\mathsf{C}_{\mathcal{L}_{q,\varepsilon}} 
\leq \mathsf{C}^{\lambda^{\star q}}_{\mathcal{L}} 
\leq \mathsf{C}_{\mathcal{L}}
\leq \mathsf{C}_{\mathcal{L}_{q,\varepsilon}} + \varepsilon ,
    \label{lem:real_to_rat_equation}
\end{align}
where the line network $\mathcal{L}_{q,\varepsilon}$ is constructed as in~\eqref{eq:rat_is_dense} and $\lambda^{\star {q}}$ is an optimal schedule of the line network $\mathcal{L}_{q,\varepsilon}$.
    \label{lem:real_to_rat}
\end{lem}

The result in Lemma~\ref{lem:real_to_rat} has the following implications:
\begin{enumerate}
\item
The statement of Lemma~\ref{lem:real_to_rat} directly implies that
$
\left | \mathsf{C}_{\mathcal{L}} - \mathsf{C}_{\mathcal{L}}^{\lambda^{\star q}} \right | \leq  \varepsilon.
$
Since this holds for all $\varepsilon > 0$, then we can use the algorithm in Section \ref{subsec:algo_integer} on $\mathcal{L}_{q,\varepsilon}$ to get a simple schedule $\lambda^{\star q,\varepsilon}$ that achieves a rate $\mathsf{C}_{\mathcal{L}}^{\lambda^{\star q}}$ that is at most $\varepsilon$ away from $\mathsf{C}_{\mathcal{L}}$.
%
%
\item 
From Lemma~\ref{lem:real_to_rat}, we have that
\begin{align*}
{ \lim_{\varepsilon \to 0} \mathsf{C}_{\mathcal{L}_{q,\varepsilon}} \leq \mathsf{C}_\mathcal{L} \leq \lim_{\varepsilon \to 0} \left[ \mathsf{C}_{\mathcal{L}_{q,\varepsilon}} + \varepsilon \right].}
\end{align*}
 Additionally, from Section~\ref{subsection:closedform}, we know that for the network $\mathcal{L}_{q,\varepsilon}$ with rational link capacities $q_i$, we have
\begin{align*}
 \lim_{\varepsilon \to 0} \mathsf{C}_{\mathcal{L}_{q,\varepsilon}}  &=\lim_{\varepsilon \to 0}  \min_{j \in [1:N]} \left \{ \frac{q_{j}\ q_{j+1}}{q_{j}+q_{j+1}} \right \}\\ &= \min_{j \in [1:N]} \left \{ \frac{\ell_{j}\ \ell_{j+1}}{\ell_{j} + \ell_{j+1}} \right \}.
\end{align*}
These two observations imply that, for any line network $\mathcal{L}$ with real link capacities, we have
\begin{align*}
\mathsf{C}_\mathcal{L} =  \min_{j \in [1:N]} \left \{ \frac{\ell_{j}\ \ell_{j+1}}{\ell_{j}+ \ell_{j+1}} \right \}.
\end{align*}
\end{enumerate}
This concludes the proof of Theorem~\ref{thm:closed_form} for real link capacities.

\section{Conclusion}\label{sec:conc}
In this work we developed a polynomial-time algorithm for finding a simple schedule (one with at most $N+1$ active states) that achieves the approximate capacity of
the $N$-relay Gaussian HD line network. 
We characterized the rate achieved by the constructed schedule in closed form, hence providing a closed-form expression for the approximate capacity of the Gaussian HD line network.
To the best of our knowledge, this is the first work which provides a closed-form expression for the approximate capacity of an HD relay network with general number of relays and designs an efficient algorithm to find 
simple schedules which achieve it.

\appendices
\section{Proof of Lemma\ref{lem:real_to_rat}}
\label{app:ProofLemmaRealtoRat}
The second inequality in~\eqref{lem:real_to_rat_equation} is straightforward from the definition of $\mathsf{C}_{\mathcal{L}}$.
Therefore, we need to prove the first and third inequalities. 
To prove the first inequality note that, from~\eqref{eq:simplified_expression_cap}, we can upperbound $ \mathsf{C}_{\mathcal{L}_{q,\varepsilon}}$ as
\begin{align*}
{ \mathsf{C}_{\mathcal{L}_{q,\varepsilon}}} &= \min_{i \in [1:N+1]} \left( \sum_{s \in \mathcal{S}_i} \lambda^{\star q}_s\right)  q_{i} \\
&\stackrel{{\rm{(a)}}}\leq \min_{i \in [1:N+1]} \left( \sum_{s \in \mathcal{S}_i} \lambda^{\star q}_s\right)  \ell_{i} = \mathsf{C}^{\lambda^{\star q}}_{\mathcal{L}},
\end{align*}
where the inequality in $\rm{(a)}$ follows since, from the construction in~\eqref{eq:rat_is_dense}, $\forall i \in [1:N+1]$, we have $q_{i} \leq \ell_{i}$. 
This proves the first inequality. 
To prove the third inequality, we use the fact that, from the construction in~\eqref{eq:rat_is_dense}, $\forall i \in [1:N+1]$, we have $q_{i} \geq \ell_{i} - \varepsilon$. 
This implies that for any schedule $\lambda$, we have
\begin{align}
{\mathsf{C}^{\lambda}_\mcal{L}} - \varepsilon 
&=
\left[ \min_{i \in [1:N+1]} \left( \sum_{\substack{s \in \mcal{S}_i}} \lambda_s \right) \ell_{i}\right] - \varepsilon \nonumber \\
&\stackrel{{\rm{(a)}}}\leq \min_{i \in [1:N+1]} \left( \sum_{\substack{s \in \mcal{S}_i}} \lambda_s \right) (\ell_{i} - \varepsilon) \nonumber 
\\
 &\leq \min_{i \in [1:N+1]} \left( \sum_{\substack{s \in \mcal{S}_i}} \lambda_s \right) q_{i} = 
{ \mathsf{C}^{\lambda}_{\mcal{L}_{q,\varepsilon}},}
\end{align}
where the inequality in ${\rm{(a)}}$ follows since $\sum_{s \in \mcal{S}_i} \lambda_s \leq 1$. 
By letting $\lambda = \lambda^\star$, with $\lambda^\star$ being an optimal schedule for $\mathcal{L}$, we have
$
 \mathsf{C}_\mcal{L} - \varepsilon \! \leq \! \mathsf{C}^{\lambda^\star }_{\mcal{L}_{q,\varepsilon}}\leq \mathsf{C}_{\mcal{L}_{q,\varepsilon}},
$
which proves the third inequality.

\section{Proof of Theorem \ref{thm:NPhard}}
\label{app:ProofofNPhard}
In this section, we prove the result in Theorem~\ref{thm:NPhard} by proving the following relations in the following subsections.
\begin{enumerate}
    \item  We first prove that the set of fundamental
        \footnote{When states or cuts are referred to as fundamental of a certain type (e.g., maximum, minimum), we mean that they form the smallest set of that type that only depends on the network topology (i.e., relays are arranged in a line) and is independent of the actual values of the point-to-point link capacities.} states in a Gaussian HD line network is equivalent to the set of fundamental maximum cuts in a Gaussian FD line network.
    \item  We next show that the problem of finding the set of fundamental maximum cuts for an $N$-relay Gaussian FD line network is equivalent to the problem of finding subsets of non-consecutive integers in $[1:N]$.
    \item Finally, we show that the number of subsets of non-consecutive integers in $[1:N]$ is exponential in $N$.
\end{enumerate}
\label{app:ProofExpStates}

\subsection{Set of Fundamental Maximum Cuts}
\label{sec:MaxCuts}

In Section~\ref{sec:cuts} we proved that we can compute the approximate capacity $\msf{C}_\mcal{L}$ in~\eqref{eq:ApproxCap} by considering only $N+1$ cuts, which are the same that one would need to consider if the network was operating in FD.
This implies that we can write~\eqref{eq:ApproxCap} as the linear program (LP)
\begin{subequations}
\label{eq:LPsimplified}
\begin{align}
\begin{array}{lll}
    {\mathsf{C}}_{\mcal{L}} = & {\rm{maximize}} & x
\\ &{\rm{subject \ to}} &  \mathbf{1}_{{N+1}} x  \leq \mathbf{A} \lambda
\\ &{\rm{and}} &  \mathbf{1}_{2^N}^T  \lambda  = 1, \ \lambda \geq \mathbf{0}_{2^N}, \ x \geq 0,
\end{array}
\label{eq:polmnuytf}
\end{align}
where $\mathbf{A} \in\mathbb{R}^{(N+1)  \times 2^N}$ has non-negative entries
\begin{align}
\label{eq:MatrA}
[\mathbf{A}]_{i,j} = \ell'_{i,j},
\end{align}
\end{subequations}
where: 
(i) $i \in [1:N+1],\ j \in [1:2^N]$; 
(ii) $\ell'_{i,j}$ is defined in~\eqref{eq:ChannelState}.
Clearly, the LP in~\eqref{eq:LPsimplified} is feasible. 
The dual of the LP in~\eqref{eq:LPsimplified} is given by
\begin{align}
\begin{array}{lll}
    {\mathsf{C}}_\mcal{L} = & {\rm{minimize}} & y
    \\ &{\rm{subject \ to}} &  \mathbf{1}_{2^N} y  \geq \mathbf{A}^T \mathbf{v}
\\ &{\rm{and}} &  \mathbf{1}_{N+1}^T  \mathbf{v}  \geq 1, \ \mathbf{v} \geq \mathbf{0}_{N+1},
\end{array}
\label{eq:dualLP}
\end{align}
where $\mathbf{A}$ is defined in~\eqref{eq:MatrA}. 
Since the LP in~\eqref{eq:dualLP} is a minimization and the entries of $\mathbf{A}$ are non-negative, then it is not hard to see that for all optimal solutions of \eqref{eq:dualLP}, we have $\mathbf{1}_{N+1}^T  \mathbf{v} = 1$.
As a result, an optimal solution of~\eqref{eq:dualLP} is a solution of 
\begin{align}
\begin{array}{lll}
    {\mathsf{C}}_\mcal{L} = & {\rm{minimize}} & y
    \\ &{\rm{subject \ to}} &  \mathbf{1}_{2^N} y  \geq \mathbf{A}^T \mathbf{v}
    \\ &{\rm{and}} &  \mathbf{1}_{N+1}^T  \mathbf{v}  = 1, \ \mathbf{v} \geq \mathbf{0}_{N+1}.
\end{array}
\label{eq:dualLP_new}
\end{align}
Since in the LP in~\eqref{eq:dualLP} we are seeking to minimize the objective function, this implies that at least one of the constraints of the type $\mathbf{1}_{2^N} y  \geq \mathbf{A}^T \mathbf{v}$ (i.e., the maximum) is satisfied with equality.
We can interpret~\eqref{eq:dualLP_new} as the problem of finding the least maximum FD cut among a class of line networks $\mcal{L}_\mathbf{V}$ derived from the original network $\mcal{L}$, where $\mathbf{V} = \{\mathbf{v} \in \mathbb{R}^{N+1}|\ \mathbf{v} \geq \mathbf{0}, \|\mathbf{v}\|_1 =1\}$. 
For each $\mathbf{v} \in \mathbf{V}$, we define a line network $\mcal{L}_\mathbf{v} \in \mcal{L}_{\mathbf{V}}$, where the capacities are modified by $\mathbf{v}$ as $\ell_i^{(v)} = \ell_i v_i$.

Let $\mathscr{F}_\mcal{M}$ be the fundamental set of maximum cuts in a FD line network, i.e., the smallest set of cuts over which we need to search for the maximum cut in FD without explicit knowledge of the values of the link capacities or their ordering. 
It is clear from definition of $\mathscr{F}_\mcal{M}$ that it is also sufficient to find the least maximum FD cut among the class of Gaussian line networks $\mcal{L}_{\mathbf{V}}$.
{With this, the rows of $\mathbf{A}^T$ (constraints in~\eqref{eq:dualLP_new}) corresponding to $\mathscr{F}_\mcal{M}$ are sufficient to find an optimal solution in~\eqref{eq:dualLP_new}.}
As a consequence of strong duality, the dual multipliers (the states $\lambda_s$ in \eqref{eq:LPsimplified}) corresponding to the fundamental maximum cuts in $\mathscr{F}_\mcal{M}$ are also sufficient to find a schedule optimal for approximate capacity. 
We now prove that, without any knowledge of the link capacities, we need to consider the network states associated to every element of $\mathscr{F}_\mcal{M}$, i.e., considering only the network states corresponding to a subset of $\mathscr{F}_\mcal{M}$ is not sufficient to achieve the approximate capacity.
To prove that, it
suffices to provide a network example, where for each $\mcal{A} \in \mathscr{F}_\mcal{M}$ the state $s_\mcal{A}^c = \mathbbm{1}_{\mathcal{A}^c}$ is the unique optimal schedule, i.e., $\lambda_{s_{\mcal{A}^c}}=1$.
For an arbitrary $\mcal{A} \in \mathscr{F}_\mcal{M}$, define the network with the link capacities
\begin{align*}
&\ell_{i} = \left \{
\begin{array}{ll}
    1 & \text{if} \quad i \in \mcal{M}_\mcal{A}
\\
\infty & \text{otherwise}
\end{array}
\right. ,\\
&\text{where}\\
& \mcal{M}_\mcal{A} = \Big\{i \in [1{:}N{+}1] \Big| \Big.\ i \in \mcal{A}\cup \{N{+}1\},\ i-1 \in \mcal{A}^c \cup \{0\} \Big\}.
\end{align*}
From the aforementioned network construction, it is not hard to see that the unique optimal schedule (one for which $C_\mcal{L} = C^{\rm FD}_\mcal{L}$) is $s_{\mcal{A}^c} = \mathbbm{1}_{\mathcal{A}^c}$, i.e., $\lambda_{s_{\mcal{A}^c}}=1$.
Therefore, for this particular network construction, the state $s_{\mcal{A}^c}$ is necessary and hence we cannot further reduce the sufficient set to a subset of $\mathscr{F}_\mcal{M}$, i.e., we need to consider the network states corresponding to every element of $\mathscr{F}_\mcal{M}$.

This result implies that, to find the smallest set of states over which we should search for an optimal schedule for approximate capacity, we should find the set of maximum cuts in FD and then consider their dual multipliers in~\eqref{eq:LPsimplified}.
In what follows, we focus on estimating the cardinality of the set of fundamental maximum cuts $\mathscr{F}_\mcal{M}$ in a FD Gaussian line network, which, as shown above, gives the cardinality of the smallest search space for an optimal schedule.

\subsection{Finding the Set of Possible Maximum Cuts through an Equivalent Problem}
We start by introducing some definitions, which will be extensively used in the rest of this section.

\begin{defin*}
    For a set of consecutive integers $[a{:}b]$, we call $\mathcal{H}$ a ``punctured'' subset of $[a{:}b]$ if $\ \forall i,j \in \mathcal{H}$ with $i\neq j$, we have $|i{-}j| {>} 1$, i.e., $\mathcal{H}$ contains no consecutive integers of $[a{:}b]$.
\end{defin*}

\begin{defin*}
    We call $\mathcal{H}$ a ``primitive punctured'' subset of $[a:b]$ if $\mathcal{H}$ is a punctured subset of $[a:b]$ and $\forall i \in [a:b]\backslash \mathcal{H},$ $\mathcal{H}\cup\{i\}$ is not a punctured set, i.e., $\mcal{H}$ is not a subset of any other punctured subset of $[a:b]$. We denote by $\mathcal{P}(a,b)$ the collection of all primitive punctured subsets of $[a:b]$.
\end{defin*}
We now use the two above definitions to state the following lemma, which is proved in the rest of this section.
\begin{lem}
\label{lem:Equivalence}
The problem of finding the set of possible maximum cuts for a Gaussian FD line network is equivalent to the problem of finding $\mathcal{P}(1,N+1)$, i.e., the collection of primitive punctured subsets of $[1:N+1]$.
\end{lem}
\begin{proof}
We start by defining two problems, namely $\mathsf{P}_1$ and $\mathsf{P}_2$, which are important for the rest of the proof:
\begin{subequations}
\begin{align}
\mathsf{P}_1 &: \qquad  \max_{\mathcal{A}\subseteq [1:N]} g_{1}(\mathcal{A}) =\sum_{\substack{ i \in \mathcal{A} \cup \{N+1\} \\ i-1 \in \mathcal{A}^c \cup \{0\}}} \ell_{i},\\
\mathsf{P}_2 &:  \qquad \max_{\substack{\mathcal{B}\subseteq [1:N+1]\\ \mathcal{B} \ \text{is punctured}}} g_{2}(\mathcal{B})=\sum_{i \in \mathcal{B}} \ell_{i}.
\label{eq:P2}
\end{align}
\end{subequations}
Note that $\msf{P}_1$ is the problem of finding the maximum FD cut in an $N$-relay Gaussian line network.
To relate the solutions of $\msf{P}_1$ and $\msf{P}_2$, we make use of the following definition.
\begin{defin*}
Given a problem $\mathsf{P}$, we denote with $\text{suf} (\mathsf{P})$ the smallest set of feasible solutions among which an optimal solution can be found for any instance of the problem.
\end{defin*}
The proof is organized as follows:
\begin{enumerate}
\item {\bf {Step 1:}} We prove that $\mathsf{P}_1$ and $\mathsf{P}_2$ are equivalent; as a consequence, there exists a function $f$ such that $ \text{suf}\left(\mathsf{P}_1\right) = f \left( \text{suf} \left(\mathsf{P}_2 \right) \right)$.
\item {\bf {Step 2:}} Next we prove that $\text{suf} \left ( \mathsf{P}_2 \right) \subseteq \mathcal{P}(1,N+1)$, which implies that 
\begin{align*}
 \text{suf} \left ( \mathsf{P}_1 \right) \subseteq f \left(\mathcal{P}(1,N+1) \right).
\end{align*}
\item {\bf {Step 3:}} The previous step implies that the set $\mathcal{M}$ of possible maximum cuts is a subset of $f \left(\mathcal{P}(1,N+1) \right)$. We finally prove that $\mathcal{M}= f \left(\mathcal{P}(1,N+1) \right)$.
\end{enumerate}
Once proved, these steps imply that we can map the problem of finding the set of possible maximum cuts for a Gaussian FD line network to the problem of finding $\mathcal{P}(1,N+1)$. 
We prove these three steps in Appendix~\ref{app:ProofLemmaEq}.
\end{proof}

\noindent{\bf{Example.}}
Consider the Gaussian FD line network with $N=7$. 
To find the set of possible maximum cuts, according to Lemma~\ref{lem:Equivalence}, we need to find the primitive punctured subsets of $[1:8]$.
This turns out to be:
\begin{align*}
    \mathcal{P}(1,8) =& \Big\{ \{1,4,7\}, \{2,4,7\}, \{2,5,7\}, \{2,5,8\}, \{1,3,5,7\},\ \Big.\\
    &\quad \Big. \{1,3,6,8\}, \{1,4,6,8\}, \{2,4,6,8\}, \{1,3,5,8\}\Big\}.
\end{align*}
It turns out that we can retrieve the candidate maximum cuts $\mathcal{A}_i$ from $\mathcal{P}(1,8)$ as follows:
\begin{align*}
\mathcal{A}_i = \mathcal{H}_i \backslash \{8\},\ \mathcal{H}_i \in \mathcal{P}(1,8), \quad \forall i \in \left [ 1: | \mathcal{P}(1,8) | \right ].
\end{align*}
To conclude the proof of Theorem~\ref{thm:NPhard}, we need to understand how the size of $\mcal{P}(1,N+1)$ grows with $N$, which is the goal of the following subsection.

\subsection{The Size of the Collection of Primitive Punctured Subsets}
In this subsection, we prove that the size of the collection of primitive punctured subsets of $[1:N+1]$ grows exponentially in $N$.
In particular, we prove the following lemma.
\begin{lem}
\label{lem:num_prim_sets}
Let $T(N)$ be the number of primitive punctured subsets of $[1:N]$. 
Then, we have the following relation,
\begin{align*}
 T(N) = T(N-2) + T(N-3).
\end{align*}
\end{lem}
The proof of the above lemma can be found in Appendix~\ref{app:ProofSizeExp}.
\begin{rem}
    {\rm
        The result in Lemma~\ref{lem:num_prim_sets} suggests that $T(N)$ grows exponentially fast. This can be proven as follows:
\begin{align*}
    T(N) &= T(N-2) + T(N-3) \\ 
    &\geq\  2\ T(N-3) \geq\ 2^k T(N-3k)\\
    &\geq  \tfrac{T(1)}{2}\ 2^{N/3},\quad \forall N\geq 4.
\end{align*}
This implies that $T(N) = \Omega(2^{N/3})$.
}
\end{rem}
Since the number of candidate active states is equal to the number of candidate maximum cuts in FD (see the discussion in Appendix~\ref{sec:MaxCuts}) and this is equal to the number of primitive punctured subsets of $[1:N+1]$ (see Lemma~\ref{lem:Equivalence}), then the number of candidate active states grows as $\Omega(2^\frac{N}{3})$. This concludes the proof of Theorem~\ref{thm:NPhard}.

\begin{rem}
    {\rm
        Using the recurrence relation in Lemma~\ref{lem:num_prim_sets}, it is not hard to prove that in fact $T(N) = \Theta(\beta^N)$ where $\beta$ is the unique real root of the polynomial $x^3 - x - 1 = 0$. 
}
\end{rem}

\section{Proof of Lemma~\ref{lem:Equivalence}}
\label{app:ProofLemmaEq}
\noindent We here prove each of the three steps in the proof of Lemma~\ref{lem:Equivalence}.
\\
{\bf {Step 1:}} We first start by proving that any feasible solution for $\msf{P}_1$ can be transformed into a feasible solution for $\msf{P}_2$ with the same value for the objective function. i.e., $\forall \mcal{A} \in [1:N]$,
\[
    \exists \text{ punctured } \mcal{B}_\mcal{A} \in [1:N+1],\ \text{s.t.} \ g_1(\mcal{A}) = g_2(\mcal{B}_\mcal{A}).
\]
To show this, for $\mathcal{A} \subseteq [1:N]$, we simply define $\mcal{B}_\mcal{A}$ as
\begin{align}
    \label{eq:Bdef}
    \mathcal{B}_\mcal{A}= \Big\{i \in [1{:}N{+}1] \Big. \Big| i \in \mathcal{A} \cup \{N{+}1\}, i{-}1 \in \mathcal{A}^c \cup \{0\} \Big \}.
\end{align}
It is clear that $\mcal{B}_\mcal{A}$ is a punctured set as $\forall i \in \mcal{B}_\mcal{A},\ i-1 \notin \mcal{B}_\mcal{A}$.
Additionally,~\eqref{eq:Bdef} directly gives us the desired relation as
\begin{align}
    \label{eq:g1_eq_g2}
    g_1 (\mathcal{A}) &= \sum_{\substack{ i \in \mathcal{A} \cup \{N+1\} \\ i-1 \in \mathcal{A}^c \cup \{0\}}} \ell_{i} = \sum_{ i \in \mathcal{B}_\mcal{A} } \ell_{i} = g_2(\mcal{B}_\mcal{A}).
\end{align}
What remains to prove now is that any feasible solution $\mathcal{B}$ for $\mathsf{P}_2$ gives a feasible solution $\mathcal{A}_\mcal{B}$ for $\mathsf{P}_1$ and $g_{1}(\mathcal{A}_\mcal{B})=g_{2}(\mathcal{B})$.
For a punctured subset $\mathcal{B}$ of $[1:N+1]$, let 
\begin{align}
\label{eq:Adef}
\mathcal{A}_\mcal{B}= f_{\mcal{A}\mcal{B}}(\mcal{B}) = \underbrace{\Big \{ i \in [1:N] \Big | \Big.\ i > \sup(\mcal{B}) \Big\}}_{\mcal{A}_{\rm tail}} \cup \underbrace{\mathcal{B} \backslash \{N + 1 \}}_{\mcal{A}_{\rm main}}.
\end{align} 
It is not hard to see that by applying the transformation in~\eqref{eq:Bdef} on $\mcal{A}_\mcal{B}$, we get back $\mcal{B}$, i.e., $\mcal{B}_{\mcal{A}_\mcal{B}} = \mcal{B}$. This is due to the fact that applying~\eqref{eq:Bdef} removes $\mcal{A}_{\rm tail}$ which is composed of a consecutive number of integers while keeping $\mcal{A}_{\rm main}$ which, since $\mathcal{B}$ is punctured, is also punctured.
Given this, we can directly see from~\eqref{eq:g1_eq_g2} that $g_1(\mcal{A}_B) = g_2(\mcal{B}_{\mcal{A}_\mcal{B}}) = g_2(\mcal{B})$.
This concludes the proof of Step 1.

\smallskip
\noindent {\bf {Step 2:}} We prove this step by showing that, if there exists an optimal solution $\mathcal{B}^\star$ of $\mathsf{P}_2$ that is not primitive, then there also exists a primitive punctured set $\mathcal{B}^\prime$ such that
$g_2(\mathcal{B}^\star)=g_2(\mathcal{B}^\prime)$.
Since $\mathcal{B}^\star$ is not a primitive punctured set, then there exists another punctured set $\mathcal{B}^\prime$ such that $\mathcal{B}^\star \subset \mathcal{B}^\prime$ and 
\begin{align*}
g_2(\mathcal{B}^\star) = \sum_{i \in \mathcal{B}^\star} \ell_{i}
\leq
\sum_{i \in \mathcal{B}^\prime} \ell_{i} = g_2(\mathcal{B}^\prime).
\end{align*}
If we take the largest such $\mathcal{B}^\prime$ we end up with a primitive punctured set. 
However, by definition (i.e., since $\mathcal{B}^\star$ is an optimal solution) we have that $\forall \mathcal{B}$ punctured, $g_2(\mathcal{B}) \leq g_2(\mathcal{B}^\star)$. 
This shows that $g_2(\mathcal{B}^\star)=g_2(\mathcal{B}^\prime)$ and therefore, $\text{suf} \left ( \mathsf{P}_2 \right) \subseteq \mathcal{P}(1,N+1)$.  This concludes the proof of Step 2.

\smallskip
\noindent{\bf {Step 3:}} In the first two steps, we proved that $\msf{P}_1$ and $\msf{P}_2$ are equivalent and that $\text{suf} \left ( \mathsf{P}_2 \right) \subseteq \mathcal{P}(1,N+1)$. 
This implies that $\text{suf} \left ( \mathsf{P}_1 \right) \subseteq f_\mcal{AB}\left(\mathcal{P}(1,N+1)\right)$, where $f_\mcal{AB}(\cdot)$ is defined in~\eqref{eq:Adef}.
We here prove that $\text{suf} \left ( \mathsf{P}_1 \right) = f_\mcal{AB}\left(\mathcal{P}(1,N+1)\right)$.
Consider an arbitrary set $\mathcal{A} \in f_{\mcal{AB}} \left(\mathcal{P}(1,N+1) \right)$. 
To prove that $\mcal{A} \in \text{suf} \left ( \mathsf{P}_1 \right)$, it suffices to provide a network (an instance of $\msf{P}_1$) for which $\mcal{A}$ is the unique maximizer of $\msf{P}_1$.
Towards this end, for the selected $\mcal{A}$, we define $\mathcal{B}_\mcal{A}$ as in~\eqref{eq:Bdef}.
We know that $\mathcal{B}_\mcal{A}$ is a primitive punctured set and $g_1(\mathcal{A})=g_2(\mathcal{B}_\mcal{A})$.
Now consider the network with link capacities
\begin{align*}
\ell_{i} = \left \{
\begin{array}{ll}
    1 & \text{if} \ i \in \mathcal{B}_\mcal{A}
\\
0 & \text{otherwise}
\end{array}
\right. .
\end{align*}
For this network, it is not hard to see that $g_2 (\mathcal{B})= |\mathcal{B} \cap \mathcal{B}_\mcal{A}|$, for any punctured set $\mathcal{B}$.
We now want to show that $\forall \mathcal{A}^\prime \in f_{\mathcal{AB}} \left(\mathcal{P}(1,N+1) \right)\backslash \mcal{A}$, we have $g_1 (\mathcal{A}^\prime) < g_1 (\mathcal{A})$.
Let $\mathcal{B}_{\mcal{A}^\prime}$ be defined as in~\eqref{eq:Bdef}.
Again, from the proof of the previous steps the set $\mathcal{B}_{\mcal{A}^\prime}$ is primitive punctured and $g_1(\mathcal{A}^\prime)=g_2(\mathcal{B}_{\mcal{A}^\prime})$.
Moreover, since $\mathcal{B}_{\mcal{A}^\prime}$ and $\mathcal{B}_\mcal{A}$ are both primitive we have that $\mathcal{B}_{\mcal{A}^\prime} \cap \mathcal{B}_{\mathcal{A}} \subset \mathcal{B}_\mcal{A}$. 
Thus, we obtain
\begin{align*}
    g_2(\mathcal{B}_{\mcal{A}^\prime})=&|\mathcal{B}_{\mcal{A}^\prime} \cap \mathcal{B}_\mcal{A}| < |\mathcal{B}_\mcal{A}| = g_2 (\mathcal{B}_\mcal{A})\\
    &\implies g_1 (\mathcal{A}^\prime)<g_1 (\mathcal{A}).
\end{align*}
Since this is true for any arbitrary $\mcal{A} \in f_\mcal{AB}\left(\mathcal{P}(1,N+1) \right)$, then it is true $\forall \mcal{A} \in f_{\mathcal{AB}} \left(\mathcal{P}(1,N+1) \right)$.
This implies that each element in $f_\mcal{AB} \left(\mathcal{P}(1,N+1) \right) $ is a unique maximum cut for some network construction. 
Therefore, without any information about the link capacities $\ell_{i}$, we cannot further reduce the set of possible maximum cuts and thus we have $\text{suf} \left ( \mathsf{P}_1 \right) = f_\mcal{AB} \left(\mathcal{P}(1,N+1) \right)$. This concludes the proof of Step 3 and hence the proof of Lemma~\ref{lem:Equivalence}.

\section{Proof of Lemma~\ref{lem:num_prim_sets}}
\label{app:ProofSizeExp}
To compute the size of $\mathcal{P}(a,b)$, it is helpful to first prove some properties of $\mathcal{P}(a,b)$ and primitive punctured subsets that will help throughout the proof.

\begin{prope}\label{prope:min_prim}
Let $\mathcal{H}$ be a primitive punctured subset of $[a:b]$, then $\min \{\mathcal{H} \} \leq a+1$.
\end{prope}
\begin{proof}
We prove this result by contradiction. 
Assume that for some primitive punctured set $\mathcal{H}$, we have $\min \{\mathcal{H}\} \geq a+2$.
This implies that $\mathcal{H} \subset [a+2:b]$.
Let $\hat{\mathcal{H}} = \mathcal{H} \cup \{a\}$. 
Since $\mathcal{H}$ is a punctured set, then $\hat{\mathcal{H}}$ is also a punctured set because $\forall i \in \mathcal{H},\ |a-i| > 1$.
But since $\mathcal{H} \subset \hat{\mathcal{H}}$, then $\mathcal{H}$ is not a primitive punctured set, which is a contradiction.
\end{proof}
Property~\ref{prope:min_prim} implies that, for a primitive punctured subset of $[a:b]$, the minimum element is either $a$ or $a+1$.
Therefore, we can write $\mathcal{P}(a,b)$ as
\begin{align*}
    \mathcal{P}(a,b) = \mathcal{P}_1(a,b)\ \uplus\ \mathcal{P}_2(a,b),
\end{align*}
where $\mathcal{P}_1(a,b)$ (respectively, $\mathcal{P}_2(a,b)$) is the collection of primitive punctured sets with minimum element $a$ (respectively, $a+1$).
Clearly, $\mathcal{P}_1$ and $\mathcal{P}_2$ are disjoint (we use $\uplus$ to indicate that the union is over disjoint sets).
Next, we prove some properties of $\mathcal{P}_1(a,b)$ and $\mathcal{P}_2(a,b)$.

\begin{prope}\label{prope:P_2}
$\mathcal{P}_2(a,b) = \mathcal{P}_1(a+1,b)$.
\end{prope}
\begin{proof}
    Let $\mathcal{H}$ be a primitive punctured subset of $[a+1:b]$ that contains the element $a+1$. $\mathcal{H}$ is also a primitive punctured subset of $[a:b]$. This follows since we cannot add $\{a\}$ to $\mathcal{H}$ to get a larger set of non-consecutive elements. 
Therefore, $\mathcal{H} \in \mathcal{P}_1(a+1,b) \implies \mathcal{H} \in \mathcal{P}_2(a,b)$.
    The reverse implication is straightforward since, by definition, $\mathcal{P}_2(a,b)$ is a primitive punctured subset which contains the element $a+1$.
\end{proof}
For the next property, we need to define a new operation on the collection of sets. 
For a collection of sets $\mathcal{Q}$, define the operation $\{i\}\sqcup \mathcal{Q} = \left\{ \{i\} \cup \mathcal{H} \left | \right. \mathcal{H} \in \mathcal{Q} \right\}$. 
We then have the following property of $\mathcal{P}_1(a,b)$.
\begin{prope}\label{prope:P_1}
    $\mathcal{P}_1(a,b) = \{a\} \sqcup \mathcal{P}(a+2,b)$.
\end{prope}
\begin{proof}
Let $\mathcal{H}$ be a primitive punctured subset of $[a+2:b]$ and define $\hat{\mathcal{H}} = \{a\} \cup \mathcal{H}$.
    Since $\mathcal{H}$ is a primitive punctured subset of $[a+2:b]$, this means that $\nexists i \in [a+2:b]\backslash \mathcal{H}$ such that $\{i\}\cup \mathcal{H}$ is a punctured sequence of $[a+2:b]$. 
    This implies that $\nexists i \in [a:b]\backslash [\mathcal{H} \cup \{i\}]$ such that $\{i\}\cup \hat{\mathcal{H}}$ is a punctured sequence of $[a:b]$. Therefore $\hat{\mathcal{H}}$ is a primitive punctured sequence of $[a:b]$, i.e., $\hat{\mathcal{H}} \in \mathcal{P}_1(a,b)$.
To prove the reverse, consider $\widetilde{\mathcal{H}} \in \mathcal{P}_1(a,b)$. We need to prove that $\hat{\mathcal{H}} = \widetilde{\mathcal{H}}\backslash\{a\}$ is a primitive punctured subset of $[a+2:b]$. Note that the definition of primitive subset of $[a:b]$ implies that $\forall i \in [a+2:b]\backslash \widetilde{\mathcal{H}}$, $\widetilde{\mathcal{H}}\cup\{i\}$ is not a punctured set. 
    Since $a \not\in [a+2:b]$, this implies that $\forall i \in [a+2:b]\backslash \hat{\mathcal{H}}$,  $\widetilde{\mathcal{H}}\cup\{i\}$ is not a punctured set. 
    Now note that since $\widetilde{\mathcal{H}} \in \mathcal{P}_1(a,b)$ then $a+1 \not\in \widetilde{\mathcal{H}}$. Therefore, $\forall i \in[a+2:b]$ removing the element $a$ from $\widetilde{\mathcal{H}}\cup\{i\}$ does not make it a punctured set.
    We therefore conclude that, $\forall i \in [a+2:b]\backslash \hat{\mathcal{H}}$,  $\hat{\mathcal{H}}\cup\{i\}$ is not a punctured set and as a result $\hat{\mathcal{H}} = \widetilde{\mathcal{H}}\backslash\{a\}$ is a primitive punctured subset of $[a+2:b]$.
\end{proof}

We now have all the necessary tools to prove Lemma~\ref{lem:num_prim_sets}.
We obtain
\begin{align*}
    \mathcal{P}(1,N) &= \mathcal{P}_1(1,N) \uplus \mathcal{P}_2(1,N)\\
    &\stackrel{\rm{(a)}}=\mathcal{P}_1(1,N) \uplus \mathcal{P}_1(2,N) \\
    &\stackrel{\rm{(b)}}= \big[ \{1\} \sqcup \mathcal{P}(3,N)\big]\ \uplus\ \big[\{2\} \sqcup \mathcal{P}(4,N)\big] ,
\end{align*}
where the equality in ${\rm{(a)}}$ follows from Property~\ref{prope:P_2} and the equality in ${\rm{(b)}}$ follows from Property~\ref{prope:P_1}.
Now note that
\begin{align*}
    \big| \left[ \{i\} \sqcup \mathcal{P}(a,N)\right] \big| &= \big|  \mathcal{P}(a,N)\big| =  \big|  \mathcal{P}(1,N-a+1)\big|\\ 
    &= T(N-a+1),
\end{align*}
since the number of sets in each collection remain the same.
Therefore, we have 
\begin{align*}
    T(N) &= |\mathcal{P}(1,N)| \\
    &= \big|  \left[ \{1\} \sqcup \mathcal{P}(3,N)\right]\ \uplus\ \left[\{2\} \sqcup \mathcal{P}(4,N)\right] \big|\\
    &= \big|  \left[ \{1\} \sqcup \mathcal{P}(3,N)\right]\ \big| + \big| \left[\{2\} \sqcup \mathcal{P}(4,N)\right] \big|\\
    &= T(N-2) + T(N-3).
\end{align*}
This concludes the proof of Lemma~\ref{lem:num_prim_sets}.

\section{Half-Duplex Routing is NP-Hard}
\label{HD_NP_hard}

We are given a directed graph $G$ with set of vertices $G(V)$, edges $E(G)$ and a source (S) and destination (D) vertices. For a graph with $N+2$ vertices, we denote the source vertex as $v_0$ and the destination vertex as $v_{N+1}$. 
For each edge $e \in E(G)$, we have an associated edge capacity $c(e) > 0$. 
        For a path $P = v_{k_1}-v_{k_2}-v_{k_3}- \dots - v_{k_{m+1}}$ of length $m$ in $G$, the Half-Duplex capacity is defined as 
        \begin{align}
            \mathsf{C}_P = \min_{i\in[2:m]}\left\{\frac{ c(e_{k_{i-1}k_{i}})\ c(e_{k_{i}k_{i+1}})}{ c(e_{k_{i-1}k_{i}})+ c(e_{k_{i}k_{i+1}}) }  \right\}.
            \label{eq:HD_capacity}
        \end{align}
        An S-D path is a path such that $v_{k_1} = v_0$ and $v_{k_m+1} = v_{N+1}$. The capacity expression in \eqref{eq:HD_capacity} can be regarded as half the minimum harmonic mean of the capacities of each two consecutive edges in the path.

        Our goal in this section, is to prove that the problem of finding the S-D simple path with the best Half-Duplex capacity in a graph is NP-Hard. Towards proving this, we first prove that the related decision problem ``HD-Path'' is NP-complete. 

\begin{defin}[HD-Path problem]
    Given directed graph $G$ and a scalar value $\mathsf{Z}$, does there exist an S-D simple path in $G$ whose Half-Duplex capacity is greater than or equal $\mathsf{Z}$ ?
\end{defin}
Since the decision problem can be reduced in polynomial time to finding the S-D simple path with the best Half-Duplex capacity, then by proving the NP-completeness of the decision problem, we also prove that the search problem is NP-Hard.

The HD-Path problem is NP because given a guess for a path, we can verify in polynomial time whether it is simple (i.e., no repeated vertices) and whether its Half-Duplex capacity is greater than or equal $Z$ by evaluating the expression in \eqref{eq:HD_capacity}.

To prove that the NP-completeness of the HD-Path problem, we are going to show that we can reduce to this decision problem from the classical 3SAT problem which is NP-complete in polynomial time.
For the {\rm 3SAT} problem, we are given a boolean expression $B$ in 3-conjunctive normal form,
\begin{align}
    B =& (p_{11} \vee p_{12} \vee p_{13}) \wedge (p_{21} \vee p_{22} \vee p_{23})\nonumber \\ 
    &\wedge \dots \wedge (p_{m1} \vee p_{m2} \vee p_{m3}).
    \label{eq:3SAT_expr}
\end{align}
$B$ is a conjuction of $m$ clauses $\{C_1,C_2,\dots,C_m\}$ which each are a disjunction of at most three literals.
A literal $p_{ij}$ is either a boolean variable $x_k$ or its negation $\bar{x}_k$.
The boolean expression $B$ is \emph{satisfiable} if the variables $\{x_k\}$ can be assigned boolean values so that $B$ is true.
The 3SAT problem answers the question: \emph{Is the given B satisfiable?}
We next prove the main result in this section through the following Lemma.
\begin{lem}
    There exists a polynomial time reduction from {\rm 3SAT} to the {\rm HD-Path} problem.
\end{lem}
\begin{proof}
    To prove this statement, we are going to create a sequence of graphs based on the Boolean statement $B$ given to the 3SAT problem.
    In each of these graphs, we shall show that the existence of a satisfying assignment for $B$ is equivalent to a particular feature in the graph.
    Our final step would give us a Half-Duplex network where the feature equivalent to a satisfying assignment of $B$ is to find a Half-Duplex path with approximate capacity greater than or equal to 1.
    Our proof follows four steps of graph constructions as follows:

\medskip

\noindent {\bf 1)}  Assume that the boolean expression $B$ is made of $m$ clauses. 
    For each clause $C_i$ in $B$, construct a gadget digraph $G_i$ with vertices $V(G_i) =\{t_i,v_{i1},v_{i2},v_{i3},r_i\}$ and edges $E(G_i) = \bigcup_{j=1}^3\left\{(t_i,v_{ij}),(v_{ij},r_i)\right\}$.
    Now we join the gadget graphs $G_i$ by adding directed edges $(r_i,t_{i+1})$, $\forall i \in [1:m{-}1]$. 
    Finally, we introduce a source vertex $S$ and destination vertex $D$ and the directed edges $(S,t_1)$ and $(r_m,D)$.
    We denote this new graph construction by $G_B$.

    \smallskip

\noindent    {\bf Example.} 
An illustration of the construction of $G_B$ for the boolean expression $B = (\bar{x}_{1} \vee x_2 \vee x_3)\wedge(x_4 \vee x_1 \vee \bar{x}_2) \wedge (\bar{x}_1 \vee x_3 \vee \bar{x}_5)$ can be seen in Fig.~\ref{fig:graph_3SAT}.

    \begin{figure}[t]
        \centering
        \includegraphics[width=0.48\textwidth]{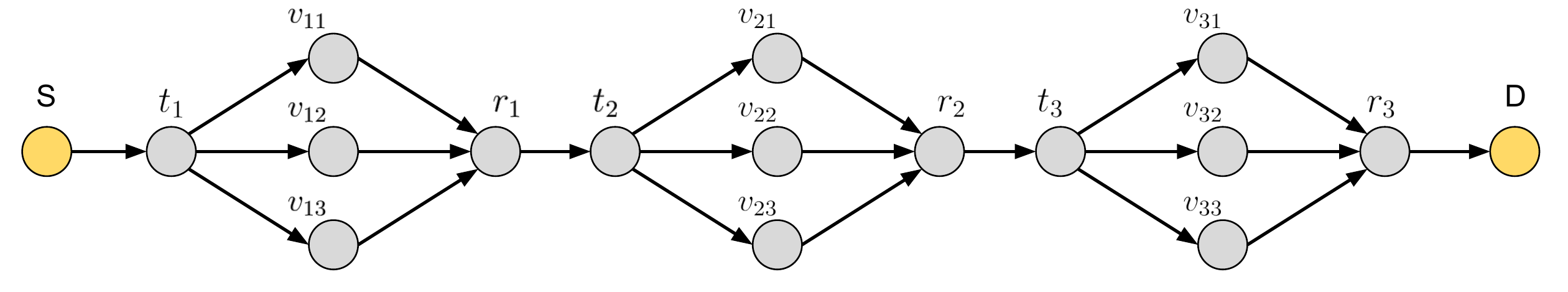}
        \caption{Graph $G_B$ constructed from boolean expression $B = (\bar{x}_{1} \vee x_2 \vee x_3)\wedge(x_4 \vee x_1 \vee \bar{x}_2) \wedge (\bar{x}_1 \vee x_3 \vee \bar{x}_5)$.}
        \label{fig:graph_3SAT}
    \end{figure}
    Note that each vertex $v_{ij}$ in $G_B$ represents a literal $p_{ij}$ in the boolean expression $B$. We call a pair of vertices $(v_{ij},v_{k\ell})$ in $G_B$, with $i < k$, as \emph{forbidden} if $p_{ij} = \widebar{p_{k\ell}}$ in $B$.
    Let $\mcal{F}$ be the set of all such forbidden pairs. 
    Consider an S-D path $P = S - t_1 - v_{1\ell_1} - r_1, t_2 - \dots - v_{m\ell_m} - r_m - D$ in the graph $G_B$ that contains at most one vertex from any forbidden pair in $\mcal{F}$. 
    Using the indexes characterizing the path $P$, if we set the literals $p_{i\ell_i}$ to be true $\forall i \in [1:m]$, then this is a valid assignment (since $P$ avoid all forbidden pairs in $\mcal{F}$). 
    Additionally, since we set one literal to be true in each clause, then this assignment satisfies $B$. Hence the existence of a path $P$ in $G_B$ that avoids forbidden pairs implies that $B$ is satisfiable.
    Similarly, we can show that if $B$ is satisfiable then we can construct a path that avoids forbidden pairs in $G_B$ using any assignment that satisfies $B$.

\medskip

\noindent {\bf 2)}    Next we modify the set of forbidden pairs $\mcal{F}$ and the graph $G_B$ such that each vertex appears at most once in $\mcal{F}$.
    For each vertex $v_{ij}$ that appears in at least one forbidden pair of $\mcal{F}$, define $V_{\mcal{F}}(v_{ij}) =\{v_{i'j'} \in V(G_B) | (v_{ij},v_{i'j'}) \in \mcal{F} \}$ and replace the vertex $v_{ij}$ in ${G}_B$ with a path of the vertices $v_{ij,k\ell}$, $\forall v_{k\ell} \in V_{\mcal{F}}(v_{ij})$. 
    We denote this new graph as $G_B^{\circ}$.
The new set of forbidden pairs $\mcal{F}^\circ$ is defined based on the set $\mcal{F}$ as $\mcal{F}^\circ = \left\{\left. \left(v_{ij,k\ell}, v_{k\ell,ij}\right) \right|\ (v_{ij},v_{k\ell}) \in \mcal{F}  \right\}$.
    Note that for this new set of forbidden pairs, each vertex in $G^\circ_B$ appears in at most one forbidden pair.
    Let $V_{\mcal{F}^\circ}$ be the set of vertices that appear in $\mcal{F}^\circ$.
    Then $\forall v_{ij,kl} \in V_{\mcal{F}^\circ}$, we replace $v_{ij,kl}$ with a path made of three vertices. 
    In particular, for any vertex $v_{ij,k\ell}$, we replace it with a directed path $a_{ij,kl} -  v_{ij,kl} - b_{ij,kl}$. 
    We call this new graph $G_B^{\star}$ and the forbidden pair set $\mcal{F}^\star = \mcal{F}^\circ$.
    The newly introduced vertices $a_{ij,k\ell}$ and $b_{ij,k\ell}$ are called \emph{a-type} and \emph{b-type} vertices respectively.
        \smallskip

        \noindent {\bf Example.} For our running example, the graph $G_B^\star$ and the forbidden pair set $\mcal{F}^\star$ are shown in Fig.~\ref{fig:graph_3SAT_modified1}.

        Similar to our argument earlier for $G_B$, note that a path in $G_B^\star$ that avoid forbidden pairs in $\mcal{F}^\star$ gives a valid satisfying assignment to satisfy the boolean argument $B$.
        In the reverse direction, if we have an assignment that satisfies $B$, then by taking one true literal from each clause $C_i$, we can choose $t_i-r_i$ paths that avoid forbidden pairs. By joining these paths together, we get an S-D path in $G_B^\star$ that avoids forbidden pairs.

\medskip

\noindent {\bf 3)} Our next step is to modify $G^\star$ to incorporate $\mcal{F}^\star$ directly into the structure of the graph. 
For each $(v_{ij,k\ell},v_{k\ell,ij}) \in \mcal{F}^\star$ introduce a new vertex $f_{ij,k\ell}$ to replace $v_{ij,k\ell}$ and $v_{k\ell,ij}$. 
All edges that were incident from (to) $v_{ij,k\ell}, v_{k\ell,ij}$ are now incident from (to) $f_{ij,k\ell}$.
We call these newly introduced vertices as \emph{f-type} vertices and denote this new graph as $G^\bullet_B$.
Note that in $G_B^\bullet$, we now have an incident edge from $a_{ij,k\ell}, a_{k\ell,ij}$ to $f_{ij,k\ell}$ and edges incident from $f_{ij,k\ell}$ to vertices $b_{ij,k\ell}, b_{k\ell,ij}$.
A path in $G^\star_B$ that avoids forbidden pairs in $\mcal{F}^\star$ gives as a path in $G^\bullet_B$ that follows the following rules:
\begin{enumerate}
    \item {\bf Rule 1:} If any \emph{f-type} vertex is visited, then it is visited at most once.
\item {\bf Rule 2:} If an \emph{f-type} vertex is visited then the preceding \emph{a-type} vertex and the following \emph{b-type} vertex both share the same index (i.e., we don't have $a_{ij,k\ell} - f_{ij,k\ell} - b_{k\ell,ij}$ or $a_{k\ell,ij} - f_{ij,k\ell} - b_{ij,k\ell}$ as a subset of our path in $G_B^\bullet$).
\end{enumerate}

It is not hard to see that an S-D path in $G_B^\bullet$ that abides to the two aforementioned rules can be appropriated to give a path that avoids forbidden pairs $\mcal{F}^\star$ in $G_B^\star$. 
This can be particularly seen by treating the subpath ($a_{ij,k\ell} - f_{ij,k\ell} - b_{ij,k\ell}$) in $G_B^\bullet$ as passing through ($a_{ij,k\ell} - v_{ij,k\ell} - b_{ij,k\ell}$) in $G^\star_B$ and similarly ($a_{k\ell,ij} - f_{ij,k\ell} - b_{k\ell,ij}$) for ($a_{k\ell,ij} - v_{k\ell,ij} - b_{k\ell,ij}$).
In other words, a problem of finding a path in $G_B^\star$ that avoids forbidden pairs in $\mcal{F}^\star$ is equivalent to finding a path in $G_B^\bullet$ that satisfies Rule 1 and Rule 2.

        \noindent {\bf Example.} For our running example, the graph $G_B^\bullet$ is shown in Fig.~\ref{fig:graph_3SAT_modified2}.
    \begin{figure}[t]
        \centering
        \includegraphics[width=0.48\textwidth]{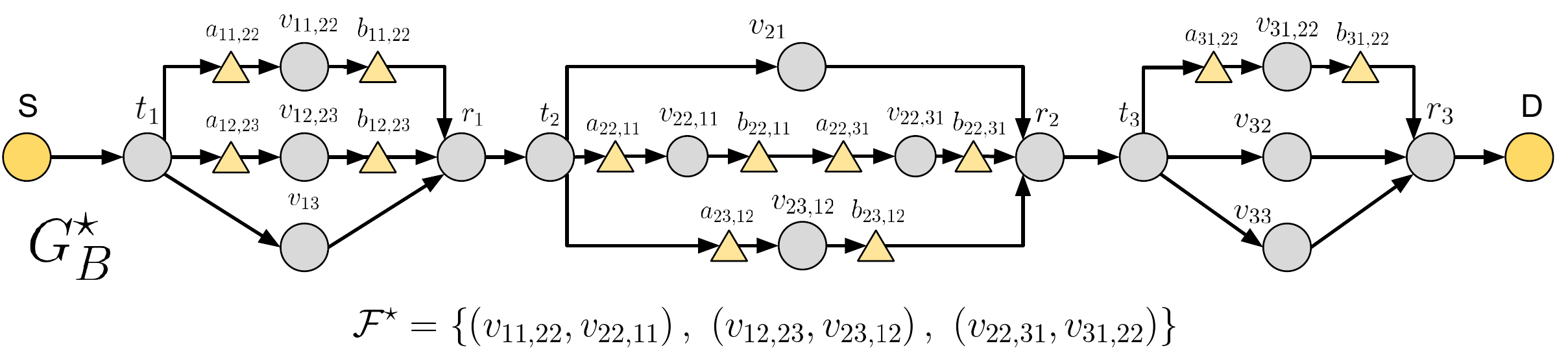}
        \caption{Graph $G^\star_B$ and the forbidden list $\mcal{F}^\star$ constructed from boolean expression $B = (\bar{x}_{1} \vee x_2 \vee x_3)\wedge$ $(x_4 \vee x_1 \vee \bar{x}_2) \wedge (\bar{x}_1 \vee x_3 \vee \bar{x}_5)$.}
        \label{fig:graph_3SAT_modified1}
    \end{figure}

\medskip

\noindent {\bf 4)} Our next step is to modify $G^\bullet_B$ by introducing edge capacities.
For any edge $e \in E(G_B^\bullet)$ that is not incident from or to an \emph{f-type} vertex, we set the capacity of that edge  $c(e) = \infty$.
For an \emph{f-type} vertex $f_{ij,k\ell}$, let $e_1,f_1$ be the edges incident to it from $a_{ij,k\ell}$ and incident from it to $b_{ij,k\ell}$. Similarly, let $e_2,f_2$ be the edges incident from $a_{k\ell,ij}$ and to $b_{k\ell,ij}$.
Then we set the edge capacities of these edges as
\begin{align*}
    c(e_1) = c(f_2) = \mathsf{Z},\quad c(f_1) = c(e_2) = \infty .
\end{align*}

    \begin{figure}[t]
        \centering
        \includegraphics[width=0.48\textwidth]{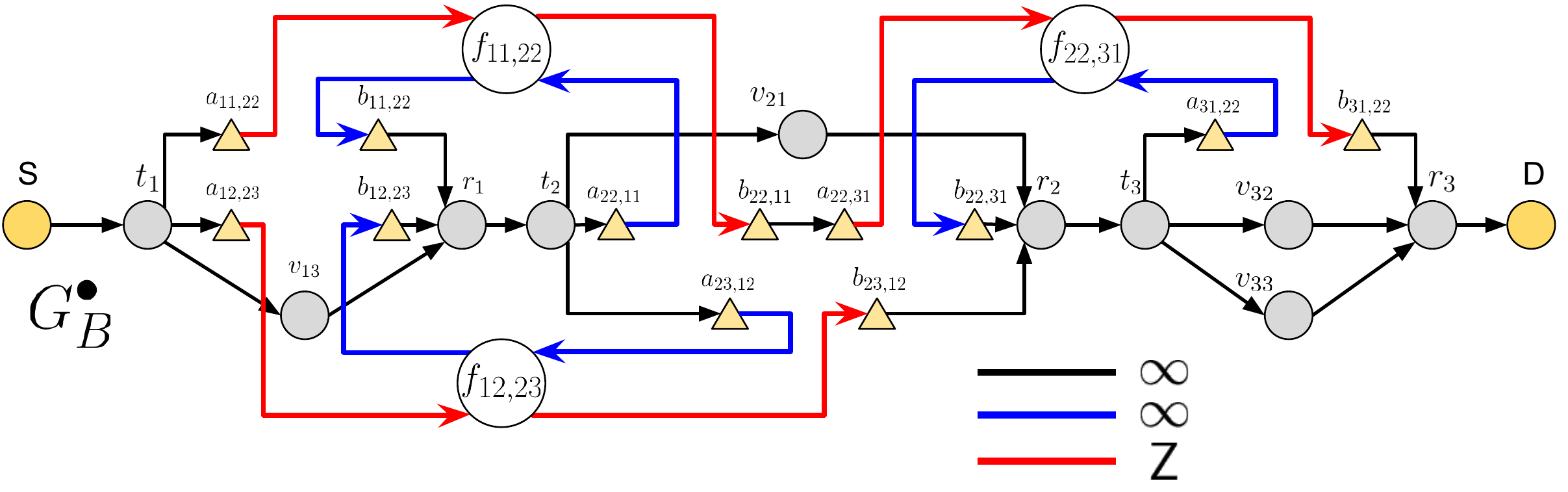}
        \caption{Graph $G^\bullet_B$ for expression $B = (\bar{x}_{1} \vee x_2 \vee x_3)\wedge$ $(x_4 \vee x_1 \vee \bar{x}_2) \wedge (\bar{x}_1 \vee x_3 \vee \bar{x}_5)$ and the associated edge capacities.}
        \label{fig:graph_3SAT_modified2}
    \end{figure}
We now need to show that finding a path satisfying Rules 1 and 2 is equivalent to finding a simple path in $G_B^\bullet$ with Half-Duplex capacity $\geq \mathsf{Z}$.
It is not hard to see that a path that follows Rules 1 and 2 is simple and has a Half-Duplex capacity $\geq \mathsf{Z}$ (by avoiding subpaths $a_{ij,k\ell} - f_{ij,k\ell} - b_{k\ell,ij}$). To prove the equivalence, we need to show that a simple path with capacity greater than or equal $\mathsf{Z}$ satisfies Rules 1 and 2. Towards this end, note that Rule 1 is inheriently satisfied since the path is simple (i.e., it visits any vertex at most once). 
For Rule 2, we argue that both subpaths are avoided by contradiction. 

Assume that the simple path selected contains a subpath of the form $a_{ij,k\ell} - f_{ij,k\ell} - b_{k\ell,ij}$. By our construction of the edge capacities, both the edges $(a_{ij,k\ell},f_{ij,k\ell})$ and $(f_{ij,k\ell},b_{k\ell,ij})$ have a capacity equal to $\mathsf{Z}$. This gives us a contradiction since half the harmonic mean between the capacities of these two consecutive edge is equal to $\mathsf{Z}/2$. Since the Half-Duplex capacity of a path is minimum of half the harmonic means of its consecutive edges, then the selected path cannot have a Half-Duplex capacity greater than or equal $\mathsf{Z}$ which is a contraction.

Now assume that the simple path selected with Half-Duplex capacity $\geq \mathsf{Z}$ contains (for some $i',j',k'$ and $\ell'$) a subpath of the form $a_{k'\ell',i'j'} - f_{i'j',k'\ell'} - b_{i'j',k'\ell'}$. 
Note that as per our construction in graph $G_B^\bullet$, we have that $i' < k'$.
Let $i^\star$ be the smallest index $i'$ for which such a subpath exists in our selected path. 
Since for the subpath in question we have that $i^\star < k'$, then to reach $a_{k'\ell',i^\star j'}$ from $S$, we already visited $r_i^\star$ earlier in the path. However to move from $b_{i^\star j',k'\ell'}$ to $D$ (after the subpath in question), we need to pass through $r_{i^\star}$ once more. This gives a contradiction with the fact that the path is simple.
This proves the fact that finding a path satisfying Rules 1 and 2 is equivalent to finding a simple path in $G_B^\bullet$ with Half-Duplex capacity $\geq \mathsf{Z}$. The second statement is an instant of the HD-Path problem.

Note that in each of the four graph constructions described earlier, we construct one from the other using a polynomial number of operations. 
Thus, this proves by construction that there exists a polynomial reduction from the 3SAT problem to the HD-Path problem.
\end{proof}

\bibliographystyle{IEEEtran}
\bibliography{LineNetwork}

\end{document}